\documentclass[11pt]{article}

\usepackage[T1]{fontenc}
\usepackage[utf8]{inputenc}
\usepackage{lmodern}
\usepackage{microtype}
\usepackage{amsmath,amssymb,amsthm,mathtools}
\usepackage{geometry}
\usepackage{hyperref}
\usepackage{enumitem}
\hypersetup{colorlinks=true,linkcolor=blue,citecolor=blue,urlcolor=blue}
\newtheorem{theorem}{Theorem}
\newtheorem{corollary}[theorem]{Corollary}
\newcommand{\T}{\mathcal T}
\newcommand{\one}{\mathbf 1}

\title{An exact Jordan signature in cumulative-count first-passage statistics}
\author{Lachlan Bridges\\
\small Independent researcher\\
\small \texttt{lachlanjbridges@gmail.com}}
\date{}

\begin{document}
\maketitle

\begin{abstract}
For cumulative-count first-passage problems, the Laplace transform of the time $\T_N$ of the $N$th count is naturally expressed through powers of a one-count kernel. A Jordan defect of that kernel does not, however, automatically survive the complete sum over terminal phases. We give an exact three-state Markov-renewal example in which it does. At $s=1$, the one-count kernel has spectrum $\{1/2,1/8,1/8\}$, with a one-dimensional eigenspace at $1/8$, and for initial phase $1$
\[
\mathbb E_1[e^{-\T_N}]
=\frac56\,2^{-N}+\frac{N+1}{6}\,8^{-N}.
\]
Thus the ordinary, unconditioned cumulative-count transform contains the explicit Jordan contribution $N/(6\,8^N)$ at every threshold $N\ge1$. For the same fixed stochastic process, the two eigenvalues meeting at $1/8$ unfold as
\[
\lambda_\pm(s)=\frac18\pm\frac1{96}\sqrt{s-1}
-\frac{131}{2304}(s-1)+O((s-1)^{3/2}),
\]
and the characteristic discriminant has a simple zero at $s=1$. We place the example beside two structures that suppress such a marginal signature: fixed post-count reset, which gives scalar renewal, and a common total holding rate, for which terminal summation scalarizes the matrix power. The model also admits a monitored-Lindblad realization, but the mechanism is entirely classical. The result provides a small exact benchmark for how generalized spectral modes can remain visible in finite-threshold first-passage statistics.
\end{abstract}

\section{Introduction}

First-passage questions for counting observables ask when an accumulated count reaches a prescribed threshold. If $\T_N$ is the time of the $N$th counted event, its Laplace transform is often most naturally organized count by count. In renewal systems this produces a scalar product of waiting-time transforms; in systems with several post-count phases it produces a matrix product. This viewpoint is standard in classical and quantum counting problems. Ptaszy\'nski \cite{ptaszynski2018} develops renewal and nonrenewal first-passage statistics, Bakewell-Smith et al. \cite{bakewell2025} use exact powered first-passage kernels for classical and quantum Markov processes, and Menczel et al. \cite{menczel2026} establish all-time relations between counting statistics and first-passage propagators. Menczel et al. also identify exceptional-point effects in first-passage statistics as a direction for investigation; the present note supplies a finite-threshold classical example in which a generalized mode survives terminal summation. Generalized Jordan chains also have a longer history in first-passage analysis of Markov additive processes \cite{dauria2010}.

Accordingly, the matrix-power identity itself is not the issue here. The question is more specific: \emph{when can a Jordan defect of a physical one-count kernel survive complete terminal summation and appear in the ordinary $N$th-count timing transform?} This distinction matters because the matrix $R_s^N$ can contain a genuine polynomial-times-exponential Jordan contribution while a particular preparation and readout annihilate it. In a cumulative-count problem the most natural readout is often exactly the terminal sum $\one$, so visibility of the generalized mode is an additional property, not a consequence of defectivity alone.

Two familiar structures make the cancellation transparent. If every counted jump resets to one fixed post-count state, the regenerative dynamics after the first count is scalar and the $N$th-count transform is a simple renewal product. If several reset phases are allowed but every phase has the same total holding rate $\gamma$, then the one-count kernel has the form $R_s=q(s)P$, with $P\one=\one$, so
\[
R_s^N\one=q(s)^N\one.
\]
Any subleading generalized eigenspace of $P$ is then removed by complete terminal summation, even though it may remain visible in terminal-resolved matrix elements. These observations make clear why an explicit marginal Jordan term is not automatic.

The example in this paper is chosen so that the cancellation does not occur. It is a three-state Markov-renewal process with phase-dependent exponential holding rates. At the distinguished Laplace point $s=1$, its one-count kernel $R_1$ has a simple eigenvalue $1/2$ and a length-two Jordan block at the subleading eigenvalue $1/8$. Starting from phase $1$, complete terminal summation gives the exact finite-threshold identity
\[
\mathbb E_1[e^{-\T_N}]
=\frac56\,2^{-N}+\frac16\,8^{-N}
+\frac{N}{6}\,8^{-N}.
\]
The last term is therefore not a terminal-conditioned effect and not an asymptotic inference: it is present explicitly in the ordinary transform for every $N\ge1$.

The same fixed model also gives a clean local analytic picture. As the Laplace variable $s$ is varied around $1$, the repeated eigenvalue $1/8$ splits in two square-root branches with leading coefficient $\pm1/96$, and the cubic characteristic discriminant has a simple zero. This provides a self-contained transform-domain defective branch point attached to the same process that exhibits the marginal Jordan term. Here $s$ is only the Laplace dual variable. Varying $s$ changes the transform weighting of one fixed stochastic dynamics; it does not tune a physical transition rate or Hamiltonian.

The paper is deliberately focused on this exact phenomenon. Section~2 defines the Markov-renewal model and its one-count Laplace kernel. Section~3 proves the defective spectrum and the terminal-summed transform and then isolates the simple algebraic reason the generalized contribution survives in this example. Section~4 gives the square-root unfolding. Section~5 compares fixed reset, common total rate, and the present state-dependent-rate construction. Section~6 embeds the process in monitored Lindblad dynamics and makes explicit that the mechanism is classical. Section~7 summarizes the resulting benchmark and its limitations. No claim is made that powered first-passage kernels, Jordan powers, or generalized Jordan-chain methods are new, and no time-domain singularity or intrinsically quantum mechanism is inferred from the transform-domain branch point.

\section{A three-state Markov-renewal kernel}

We work in time units for which the distinguished Laplace point is $s_*=1$. The post-count phase takes values in $\{1,2,3\}$. Conditional on the current phase $i$, the waiting time until the next counted event is exponential with rate $\gamma_i$; when that event occurs, the new phase is $j$ with probability $P_{ij}$. Thus the model is a finite-state Markov-renewal process whose embedded chain is $P$ and whose holding law depends on the departure phase.

Take
\[
P=
\begin{pmatrix}
13/22&7/22&1/11\\
1/4&11/20&1/5\\
1/3&1/3&1/3
\end{pmatrix},
\qquad
(\gamma_1,\gamma_2,\gamma_3)=\left(\frac{11}{13},\frac57,3\right).
\]
Each row of $P$ sums to one and each $\gamma_i$ is positive. If $\tau_n$ is the $n$th inter-count time and $X_n$ is the phase immediately after the $n$th count, then
\[
\Pr(\tau_n\in dt,\ X_n=j\mid X_{n-1}=i)
=\gamma_iP_{ij}e^{-\gamma_i t}\,dt.
\]
Counted events include self-transitions $i\to i$: such an event increments the count even though the post-count phase remains $i$. Thus $\T_N$ is the time of the $N$th counted event, not necessarily the time of the $N$th change of phase. The physical $N$th-count time is
\[
\T_N=\tau_1+\cdots+\tau_N.
\]

Laplace transforming one count interval gives
\[
(R_s)_{ij}
=\int_0^\infty e^{-st}\gamma_iP_{ij}e^{-\gamma_i t}\,dt
=\frac{\gamma_i}{\gamma_i+s}P_{ij},
\]
and hence
\begin{equation}
\label{eq:kernel}
R_s=\operatorname{diag}\!\left(
\frac{\gamma_1}{\gamma_1+s},
\frac{\gamma_2}{\gamma_2+s},
\frac{\gamma_3}{\gamma_3+s}
\right)P,
\qquad \Re s>-\frac57.
\end{equation}
For real $s\ge0$, $R_s$ is nonnegative; it is stochastic at $s=0$ and strictly substochastic for $s>0$. The unequal factors $\gamma_i/(\gamma_i+s)$ are the only structural departure from a common-rate reset chain that will matter below.

Let $e_i$ be the $i$th coordinate vector and let $\one=(1,1,1)^T$. Repeated conditioning at count times gives the standard Markov-renewal product. After summing over the terminal phase,
\begin{equation}
\label{eq:terminal-sum}
\boxed{\;
\mathbb E_i[e^{-s\T_N}]
=e_i^TR_s^N\one.
\;}
\end{equation}
The right vector $\one$ is therefore not an arbitrary test vector in the present problem: it represents the physically ordinary choice in which the terminal reset phase is not observed. The question addressed below is precisely whether the generalized eigenspace of $R_s$ survives this natural summation.

For later comparison, note that
\[
R_s\one=
\begin{pmatrix}
\gamma_1/(\gamma_1+s)\\
\gamma_2/(\gamma_2+s)\\
\gamma_3/(\gamma_3+s)
\end{pmatrix}.
\]
With a common holding rate these three entries would coincide and $\one$ would be a right eigenvector for every $s$. In the present model they are generally distinct. This fact alone does not guarantee a visible Jordan term, but it removes the automatic scalarization that occurs in the common-rate case.

\section{Exact terminal-summed Jordan signature}

\begin{theorem}
At $s=1$,
\[
R_1=
\begin{pmatrix}
13/48&7/48&1/24\\
5/48&11/48&1/12\\
1/4&1/4&1/4
\end{pmatrix}
=:R_*.
\]
Its characteristic polynomial is
\[
\det(\lambda I-R_*)
=\left(\lambda-\frac12\right)
 \left(\lambda-\frac18\right)^2,
\]
and the eigenspace at $1/8$ is one-dimensional. Hence $R_*$ has a length-two Jordan block at the subleading eigenvalue $1/8$.

For initial phase $1$ and every $N\ge1$,
\begin{equation}
\label{eq:main-transform}
\boxed{\;
\mathbb E_1[e^{-\T_N}]
=e_1^TR_*^N\one
=\frac56\,2^{-N}+\frac{N+1}{6}\,8^{-N}.
\;}
\end{equation}
In particular, the ordinary terminal-summed first-passage transform contains the nonzero Jordan contribution
\[
\frac{N}{6}\,8^{-N}.
\]
\end{theorem}

\begin{proof}
At $s=1$ the three row factors are
\[
\frac{\gamma_1}{\gamma_1+1}=\frac{11}{24},
\qquad
\frac{\gamma_2}{\gamma_2+1}=\frac5{12},
\qquad
\frac{\gamma_3}{\gamma_3+1}=\frac34.
\]
Multiplying these factors into the corresponding rows of $P$ gives the displayed matrix $R_*$ entry by entry.

A direct determinant expansion gives
\[
\det(\lambda I-R_*)
=\lambda^3-\frac34\lambda^2+\frac9{64}\lambda-\frac1{128}
=\left(\lambda-\frac12\right)\left(\lambda-\frac18\right)^2.
\]
Moreover,
\[
48\left(R_*-\frac18I\right)
=
\begin{pmatrix}
7&7&2\\
5&5&4\\
12&12&6
\end{pmatrix}.
\]
The first and third columns are linearly independent, while the first two coincide, so this matrix has rank $2$. Therefore
\[
\dim\ker\!\left(R_*-\frac18I\right)=1.
\]
For an explicit Jordan chain,
\[
v=(-1,1,0)^T,\qquad
w=(-16,0,32)^T
\]
satisfy
\[
\left(R_*-\frac18I\right)v=0,
\qquad
\left(R_*-\frac18I\right)w=v.
\]
Thus the repeated eigenvalue $1/8$ is genuinely defective.

Set
\[
a_N=e_1^TR_*^N\one.
\]
Since the minimal polynomial is
\[
\left(x-\frac12\right)\left(x-\frac18\right)^2,
\]
the scalar sequence has the form
\[
a_N=A\,2^{-N}+(B+CN)\,8^{-N}.
\]
Three exact values determine the coefficients. First $a_0=1$. Next, the first row sum of $R_*$ gives
\[
a_1=\frac{11}{24}.
\]
Also
\[
R_*\one=
\begin{pmatrix}
11/24\\[2pt]5/12\\[2pt]3/4
\end{pmatrix},
\]
and hence
\[
a_2
=\frac{13}{48}\frac{11}{24}
+\frac7{48}\frac5{12}
+\frac1{24}\frac34
=\frac{83}{384}.
\]
Solving
\[
A+B=1,
\qquad
\frac A2+\frac{B+C}{8}=\frac{11}{24},
\qquad
\frac A4+\frac{B+2C}{64}=\frac{83}{384}
\]
gives
\[
A=\frac56,\qquad B=C=\frac16.
\]
Therefore
\[
a_N=\frac56\,2^{-N}+\frac{N+1}{6}\,8^{-N}.
\]
Finally,
\[
\frac{N+1}{6}\,8^{-N}
=\frac16\,8^{-N}+\frac{N}{6}\,8^{-N},
\]
so the coefficient multiplying the nilpotent Jordan contribution is nonzero after complete terminal summation.
\end{proof}

\subsection*{Why terminal summation does not cancel the generalized mode}

The theorem can be read directly at the level of the physical observable $e_1^TR_*^N\one$. A length-two Jordan block permits a contribution proportional to $N(1/8)^N$, but whether it appears in this scalar quantity depends on the preparation $e_1$ and the terminal vector $\one$. Here the exact three-value calculation above gives $C=1/6$, so the generalized component has nonzero overlap with both ends of the observable.

There is also a useful structural way to see why cancellation is not forced. At $s=1$,
\[
R_*\one=
\begin{pmatrix}
11/24\\[2pt]5/12\\[2pt]3/4
\end{pmatrix},
\]
which is not proportional to $\one$. Complete terminal summation therefore does not collapse all powers $R_*^N\one$ onto a single scalar mode. In the common-total-rate comparison of Section~5, by contrast, $\one$ is a right eigenvector for every $s$, and this collapse occurs identically. The present calculation makes only the example-specific conclusion needed here: for this state-dependent-rate kernel the generalized $1/8$ mode survives the natural terminal sum with coefficient $1/6$.

The result is an exact finite-threshold statement rather than a large-$N$ approximation. The Perron contribution $(5/6)2^{-N}$ dominates asymptotically, as expected for a positive kernel, while the Jordan term is subleading. Nevertheless, its polynomial prefactor is present for every $N$, and its coefficient can be checked directly from the exact transform values. This separation between asymptotic dominance and exact finite-$N$ visibility is one reason the example is useful as a benchmark.

\section{Local square-root unfolding}

The Jordan identity at $s=1$ is already enough for the finite-threshold statement. It is nevertheless useful to examine how the repeated eigenvalue sits inside the analytic family $R_s$. Because $R_s$ is obtained by Laplace transforming one fixed process, this is a local spectral question in the transform variable, not a perturbation of the underlying dynamics.

\begin{corollary}
Let $t=s-1$. The two eigenvalue branches of $R_s$ that meet at $1/8$ satisfy
\begin{equation}
\label{eq:puiseux}
\boxed{\;
\lambda_\pm(s)
=\frac18\pm\frac1{96}t^{1/2}
-\frac{131}{2304}t
+O(t^{3/2}).
\;}
\end{equation}
The characteristic discriminant has a simple zero at $s=1$. Thus $(s,\lambda)=(1,1/8)$ is a simple ramification point of the spectral curve over the $s$-plane, and the two eigenvalue branches have an order-two square-root unfolding there.
\end{corollary}

\begin{proof}
Clearing the nonzero denominators of $\det(\lambda I-R_s)=0$ gives
\begin{equation}
\label{eq:cleared-cubic}
\begin{aligned}
F(s,\lambda)
={}&4(s+3)(7s+5)(13s+11)\lambda^3\\
&-(689s^2+1794s+973)\lambda^2\\
&+(299s+349)\lambda-36.
\end{aligned}
\end{equation}
At $s=1$,
\[
F(1,\lambda)=36(2\lambda-1)(8\lambda-1)^2.
\]
Write
\[
\lambda=\frac18+\mu,\qquad s=1+t.
\]
To establish the existence of the local branches before computing their coefficients, set
\[
t=z^2,\qquad \mu=zu,
\]
and define, using the removable analytic extension at $z=0$,
\[
H(z,u)=\frac{F(1+z^2,\,1/8+zu)}{z^2}.
\]
The Taylor expansion of $F$ at $(1,1/8)$ gives
\[
H(0,u)=\frac3{16}-1728u^2.
\]
Its two roots are $u_0=\pm1/96$, and
\[
\partial_uH(0,u_0)=-3456u_0=\mp36\ne0.
\]
The analytic implicit-function theorem therefore gives two analytic functions $u_\pm(z)$ near $z=0$, with $u_\pm(0)=\pm1/96$ and $H(z,u_\pm(z))=0$. Hence
\[
\lambda_\pm(s)=\frac18+z\,u_\pm(z),\qquad z^2=s-1,
\]
are genuine local Puiseux branches.

To determine the next coefficient, the terms needed through order $t^{3/2}$ are
\[
F(1+t,\tfrac18+\mu)
=\frac3{16}t-1728\mu^2-197t\mu+4608\mu^3
+O(t^2,t\mu^2).
\]
Seeking
\[
\mu=a\,t^{1/2}+b\,t+O(t^{3/2})
\]
and matching the coefficients of $t$ and $t^{3/2}$ gives
\[
\frac3{16}-1728a^2=0,
\]
and
\[
-3456ab-197a+4608a^3=0.
\]
Hence
\[
a=\pm\frac1{96},
\qquad
b=-\frac{131}{2304},
\]
which proves the stated branches.

For completeness, the discriminant $\Delta(s)$ of the cleared cubic factors as
\[
\Delta(s)=(s-1)G(s),
\]
with
\[
G(1)=3869835264\neq0.
\]
The discriminant zero at $s=1$ is therefore simple. For real $s>1$ sufficiently close to $1$, the leading square-root term is real and the two nearby eigenvalues are distinct and real. For real $s<1$ sufficiently close to $1$, they form a complex-conjugate pair.
\end{proof}

Nothing in this calculation requires interpreting $s$ as a tunable physical parameter. The matrices $R_s$ are Laplace transforms of the same fixed inter-count law. For every fixed threshold $N$ and initial phase $i$, the quantity
\[
\phi_{N,i}(s):=e_i^TR_s^N\one
\]
is a rational function of $s$, with possible poles only at $-11/13$, $-5/7$, and $-3$. Hence both $R_s$ and the finite-threshold transform $\phi_{N,i}(s)$ are analytic in a neighbourhood of $s=1$. The square-root branching in \eqref{eq:puiseux} concerns the individual eigenvalue parametrizations, not a singularity of $R_s$ or of the finite-threshold first-passage transform. In particular, it does not by itself assert a singularity of the time-domain density, an inverse-Laplace asymptotic governed by $s=1$, or robustness under arbitrary perturbations of the stochastic model. Those are separate questions.

\section{Structural comparison}

The terminal-summed Jordan term becomes clearer when placed beside two simpler reset structures. The comparison is not intended as a classification theorem. It isolates three concrete mechanisms: scalar renewal after a fixed reset; automatic terminal-sum scalarization under a common total rate; and the present example, where phase-dependent Laplace row factors leave room for a generalized mode to contribute to the terminal sum.

\paragraph{Fixed post-count reset: scalar ladder.}
Suppose every counted jump resets the system to one fixed state $\sigma$. After the first count there is only one regenerative phase. If $w_{X_0}(s)$ is the first waiting-time transform from the initial state and $w(s)$ is the common transform after reset, then
\[
\mathbb E_{X_0}[e^{-s\T_N}]
=w_{X_0}(s)\,w(s)^{N-1}.
\]
Equivalently, the canonical post-count ladder is the $1\times1$ matrix $[w(s)]$. A nontrivial factor $N\lambda^N$ therefore cannot arise from defectivity of this cumulative-count ladder. This is simply the spectral form of scalar renewal, not a new renewal theorem.

\paragraph{Several reset phases with one common total rate: terminal summation scalarizes.}
Let the post-count phase evolve according to a row-stochastic matrix $P$, but let every phase have the same exponential holding rate $\gamma$. Then
\[
R_s=q(s)P,\qquad q(s)=\frac{\gamma}{\gamma+s},
\qquad P\one=\one.
\]
Consequently
\[
R_s^N\one=q(s)^NP^N\one=q(s)^N\one,
\]
and for every initial phase $i$,
\[
e_i^TR_s^N\one=q(s)^N.
\]
Thus a defective subleading mode of $P$ may be visible in a terminal-resolved matrix element of $R_s^N$, but complete terminal summation removes it. The cancellation is exact because $\one$ is a right eigenvector of $R_s$.

\paragraph{The present state-dependent-rate example: nonzero generalized overlap.}
Here
\[
R_s=D_sP,\qquad
D_s=\operatorname{diag}\!\left(
\frac{\gamma_1}{\gamma_1+s},
\frac{\gamma_2}{\gamma_2+s},
\frac{\gamma_3}{\gamma_3+s}\right),
\]
with unequal diagonal entries. At $s=1$, the row sums are
\[
\frac{11}{24},\qquad \frac5{12},\qquad \frac34,
\]
so $\one$ is not a right eigenvector of $R_*$. The exact computation of Section~3 then shows that the generalized $1/8$ mode has nonzero preparation/readout overlap and contributes
\[
\frac{N}{6}\,8^{-N}
\]
to the ordinary terminal-summed transform. The point is not that state-dependent rates are the only possible route to such visibility; rather, in this explicit model they remove the common-rate scalarization and yield a nonzero generalized-mode coefficient.

These three cases distinguish defectivity from observability. A matrix may possess a Jordan block, yet the cumulative-count statistic can remain purely scalar because of the reset structure or because the terminal vector projects away the subleading modes. Conversely, Theorem~1 shows that complete terminal summation need not erase the polynomial prefactor. For applications, this suggests checking the actual prepared-and-summed quantity $e_i^TR_s^N\one$, rather than inferring an observable signature from the spectrum of $R_s$ alone.

\section{Monitored-Lindblad realization}

Although the construction has been presented as a classical Markov-renewal process, it also sits inside the standard monitored-Lindblad framework for counted quantum jumps. This realization is useful because it shows that the exact first-passage kernel and its marginal Jordan term are compatible with a legitimate open-system counting model, while at the same time making clear that no quantum coherence is responsible for the effect.

\begin{corollary}[Monitored-Lindblad realization]
Let $\{|1\rangle,|2\rangle,|3\rangle\}$ be an orthonormal basis, set $H=0$, and monitor every jump operator
\[
L_{ji}=\sqrt{\gamma_iP_{ij}}\,|j\rangle\langle i|,
\qquad i,j\in\{1,2,3\}.
\]
All nine channels are monitored and counted, including the diagonal channels $L_{ii}$: such an event increments the count even though a trajectory already in $|i\rangle$ remains in the same normalized state after the event. Starting immediately after a monitored jump in state $|i\rangle$, the next monitored jump occurs after an exponential waiting time of rate $\gamma_i$ and resets the system to $|j\rangle$ with probability $P_{ij}$. Consequently the count-to-count Laplace kernel is exactly the matrix $R_s$ of Section~2; in particular, for initial state $|1\rangle$ and $s=1$, the terminal-summed $N$th-count transform is the expression in Theorem~1.
\end{corollary}

\begin{proof}
Because $P$ is row stochastic,
\[
\sum_j L_{ji}^\dagger L_{ji}
=\gamma_i|i\rangle\langle i|.
\]
With $H=0$, a trajectory starting from $|i\rangle$ has no-count survival probability $e^{-\gamma_i t}$. The probability density that the next monitored jump occurs at time $t$ through channel $i\to j$ is therefore
\[
\gamma_iP_{ij}e^{-\gamma_i t}.
\]
After that jump the normalized state is exactly $|j\rangle$. Laplace transformation gives
\[
(R_s)_{ij}=\frac{\gamma_i}{\gamma_i+s}P_{ij},
\]
and repeated monitored jumps reproduce the Markov-renewal product of Section~2.
\end{proof}

The mechanism is nevertheless classical. After every monitored event the trajectory lies in one of the three reset-basis states, and there is no Hamiltonian evolution mixing those states. The phase dependence enters only through the classical holding rates and reset probabilities. In particular, the surviving $N8^{-N}$ term should not be described as coherence-driven or as evidence that noncommutativity is required. The Lindblad formulation is a realization of the same Markov-renewal process, not a different mechanism.

Likewise, the square-root eigenvalue splitting belongs to the analytic family of Laplace kernels for this fixed dynamics. The distinguished value $s=1$ is not a laboratory control setting. The exact physical statement is that the cumulative-count timing transform at that Laplace value contains the Jordan contribution displayed in Theorem~1; the nearby square-root branches describe the local analytic structure of the transform kernel.

\section{Discussion}

The main point of this paper is a small but concrete distinction between spectral defectivity and a directly observed cumulative-count signature. In the three-state process above, $R_1$ has a genuine length-two Jordan block at $1/8$, and the physically ordinary terminal-summed transform does not annihilate its nilpotent contribution. Instead,
\[
\mathbb E_1[e^{-\T_N}]
=\frac56\,2^{-N}
+\frac16\,8^{-N}
+\frac{N}{6}\,8^{-N}.
\]
All three terms are exact. The polynomial prefactor is subleading, but it is present at finite threshold without conditioning on the terminal phase.

The structural comparison explains why this is worth stating explicitly. Fixed post-count reset removes matrix defectivity from the count ladder altogether. A common total rate can leave a defective multi-phase ladder while making $\one$ a right eigenvector, so complete terminal summation sees only the scalar factor $q(s)^N$. The present phase-dependent rates do not impose that scalarization, and the exact coefficient calculation shows that the generalized mode is retained. None of these observations requires broad claims about the uniqueness or minimality of the mechanism.

The example can therefore serve as an exact benchmark for analytical or numerical treatments of first-passage kernels near non-diagonalizable points. A method that resolves only eigenvalues should recover the repeated value $1/8$ but may miss the polynomial factor. A method that resolves the full matrix power but reports only terminal-resolved entries may miss the additional question of whether the generalized mode survives marginalization. Here both checks are available in closed form: the finite-$N$ transform fixes the Jordan coefficient, while the local Puiseux expansion fixes the square-root splitting of the same kernel family.

A natural next problem is to characterize, for broader Markov-renewal families, when preparation and terminal summation annihilate or retain generalized eigenspaces of the one-count kernel. The present note does not attempt such a classification. Its role is more limited: it provides one exact positive example, two transparent cancellation mechanisms, and a monitored-Lindblad realization in which all calculations can be checked directly. That makes the visibility question concrete without conflating it with the already-standard matrix-power formalism.

In summary, a Jordan defect of a one-count first-passage kernel can survive the most natural terminal marginalization. For the explicit three-state state-dependent-rate process constructed here, the surviving term is exactly $N/(6\,8^N)$, and the repeated eigenvalue sits at a simple square-root branch point of the analytic Laplace-kernel family. The example is classical in mechanism, exact at finite threshold, and sufficiently simple to function as a reference calculation for future work on spectral structure in cumulative-count first-passage statistics.

\end{document}